%% file: P1B_journal.tex
\documentclass[12pt, draftclsnofoot, onecolumn]{IEEEtran}
\title{Energy Consumption Optimization in Mobile Communication Networks}
\author{Florian Bahlke and Marius Pesavento}
\date{\today}

\usepackage[T1]{fontenc}
\usepackage{amsmath,graphicx,mathtools}
\usepackage{tikz}
\usepackage{scalerel,amssymb}
\usepackage{tabularx}
\usepackage{algpseudocode}
\usepackage{url}
\usepackage{bm}
\usepackage{amsthm}
\usepackage{algorithm}
\usepackage{cite}
\usepackage{tikz}
\usepackage{pgfplots}
\usetikzlibrary{plotmarks}
\usetikzlibrary{calc,positioning,arrows.meta,backgrounds}

\newcommand{\bs}{\boldsymbol}

\DeclareMathOperator*{\argmax}{arg\,max}

\newtheorem{proposition}{Proposition}

\begin{document}

\maketitle

\begin{abstract}
This work addresses the challenge of minimizing the energy consumption of a wireless communication network by joint optimization of the base station transmit power and the cell activity. A mixed-integer nonlinear optimization problem is formulated, for which a computationally tractable linear inner approximation algorithm is provided. The proposed method offers great flexibility in optimizing the network operation by considering multiple system parameters jointly, which mitigates a major drawback of existing state-of-the-art schemes that are mostly based on heuristics. Simulation results show that the proposed method exhibits high performance in decreasing the energy consumption, and provides implicit load balancing in difficult high demand scenarios.
\end{abstract}

\section{Introduction} \label{sec:intro}
In the evolution of wireless communication networks over the recent years, new technologies have been proposed to fulfill the increasing performance requirements for the upcoming fifth generation (5G) and subsequent generations of mobile communication networks \cite{andrews14,boccardi14,iwamura15}. In addition to serving users with ever increasing data rates, novel applications require very low latency connections and extreme reliability. Among the most promising technologies for 5G are Massive-MIMO systems, Millimeter-Wave communications and heterogeneous network structures with dense cell deployments. Between these three options, the latter one arguably poses the lowest technological risks and level of commitment from network providers. For such dense and heterogeneous networks, the existing cell architecture is supplemented with additional cells containing base stations of variable size, both in transmit power and coverage area. This densification of the network has been identified as a promising and scalable approach for the next decades of wireless communications \cite{bhushan14,ge16,heath17}. Due to the increase in intercell interferences limiting the achievable data throughput, novel control schemes for such networks need to be devised that supersede the established strategy of deploying additional cells without increasing the amount of coordination between them \cite{andrews16,cavalcante14,rowell14}. The wireless communication networks of the future are envisioned to have a significantly higher energy efficiency in terms of energy consumption per transmitted bit of data. In the 5G standard, this will be achieved trough intelligent switching of each cell's operation between active phases and sleep modes - abandoning the always-on and always-connected concept of contemporary base stations - a dynamic scaling of the transmit power, and an energy-focused design of multi-antenna systems \cite{lange14,cavalcante14,correia10,vereecken11}. \\
We propose a method for minimizing the energy consumption of the wireless communication network, subject to cell load constraints that prevent cells from being unable to serve the demand of associated users with their available time-frequency resources. This approach is suitable for the planning the network parameters ahead of operation, and complements energy efficient transceiver techniques commonly applied in-operation for example to maximize instantaneous data rates. In previous research, extensive effort has been invested into the analysis and optimization of cell loads for heterogeneous mobile communication networks \cite{yuan17,yuan12,siomina13,bahlke15,bahlke18,yangz15}. The cell load has been used in various schemes to optimize the transmit powers \cite{yuan15,yuan16}, in the design of energy-efficient beamformers for multi-antenna systems \cite{chen16,miao11}, and to optimize the cell on-off status to enable scheduling for sleep mode and activity periods \cite{celebi17,lei15}. These methods share one fundamental disadvantage, which is that they cannot jointly optimize the transmit power and the cell activity status. Switching cells off is just considered implicitly, as the transmit power being scaled down to zero \cite{yuan15,utschick16}. The transmit power in a practical system however might be lower-bounded by a nonzero level, for example due to transmit power independent losses and nonlinearities in the power amplifiers \cite{kandukuri02,desset12,arnold10}. Heuristic approaches also heavily rely on the cell load being a strictly decreasing continuous function of its transmit power, which requires multiple simplifications in the way how the network is modeled, particularly regarding the used adaptive modulation and coding schemes. For example, the load a user adds to a cell needs to be a strictly decreasing function of the user's signal-to-interference-and-noise ratio (SINR), the assignment of users to cells needs to be constant and the operable transmit power range needs to be lower-bound by zero, which all do not necessarily apply in practical systems.   \\
In this work, we propose an approach for minimizing the network's energy consumption based on Mixed-Integer Linear Programming (MILP), which expands upon state of the art solutions in the following aspects:
\begin{itemize}
	\item The transmit power and the activity status of the cell (on or off) are jointly optimized. This leads to a mixed-integer problem as, e.g.~the transmit power is optimized on a continuous scale and the cell activity indicator is binary.
	\item While the original optimization problem is nonlinear and computationally intractable to solve, we propose a linear inner approximation. The solution of this approximate problem is always feasible for the original problem.
	\item The proposed method easily incorporates additional convex constraints such as minimum transmit power and minimum SINR threshold constraints as well as upper bounds on the user rates due to finite modulation and coding schemes. 
	\item The assignment of users to cells is one of the design parameters, and changes dynamically according to which cell provides the strongest signal. The proposed scheme also allows the incorporation of other user allocation rules.
\end{itemize}
The remainder of the paper is structured as follows: In Section \ref{sec:sysmod} we introduce the system model for the wireless communication network. A mixed-integer nonlinear programming (MINLP) approach to minimize the network's energy consumption is introduced in Section \ref{sec:opt}, for which we provide an inner linear approximation (MILP). Simulation results for different energy consumption models and a comparative analysis between the proposed and alternative methods are provided in Section \ref{sec:simres}. Finally, we summarize the results and provide an outlook onto future work in Section \ref{sec:conc}.

\section{System Model} \label{sec:sysmod}
We consider a downlink wireless communication network with $K$ cells, each equipped with a single antenna base station (BS). The transmit power of the BS in cell $k=1,\ldots,K$ is denoted as $p_k$, and the vector of all transmit powers as $\bs{p} = \left[p_1, p_2, \ldots, p_K\right]^T$. In the following we assume that each cell is defined as the coverage area of the BS, and therefore we use the terms interchangeably. In practical networks the transmit power $p_k$ is generally confined to lie in a the interval 
\begin{equation} \label{eq:pconst}
0 < P_k^{\mathrm{MIN}} \leq p_k \leq P_k^{\mathrm{MAX}},
\end{equation}
where, due to physical hardware limitations, such as linearity constraints in the power amplifiers and radiation efficiency requirements of the antenna the thresholds $P_k^{\mathrm{MIN}}$ and $P_k^{\mathrm{MAX}}$ are positive (excluding $P_k^{\mathrm{MIN}}=0$) and finite \cite{kandukuri02,desset12,arnold10}. The network contains $M$ demand points (DP) with DP $m=1,\ldots,M$ having the data demand $D_m$. A DP may represent a single mobile node, or in case of a cluster of closely spaced mobile nodes with similar channel characteristics to the connected BS, the accumulation of multiple nodes. The gain of BS $k$ is denoted by $\tilde{g}^\text{BS}_k$, correspondingly we use $\tilde{g}^\text{DP}_m$ as the antenna gain of the DP $m$. The large-scale path attenuation factor of signals transmitted from the BS in cell $k$ to DP $m$ are denoted as $\tilde{g}^\text{PATH}_{mk}$. Small scale fading parameters will be neglected because the proposed method is assumed to be applied on a network planning timescale, and on averaged rather than instantaneous channel information. Combining the aforementioned factors, we denote as $g_{km} = \tilde{g}^\text{BS}_k \tilde{g}^\text{PATH}_{mk} \tilde{g}^\text{DP}_m$ the attenuation factor for transmissions between BS $k$ and DP $m$. \\
The SINR of cell $k$ serving DP $m$ can be modeled as
\begin{equation} \label{eq:sinr}
\gamma_{km}=\frac{p_k g_{km}}{\sum_{j=1,\ldots,K \backslash \{k\}} p_j g_{jm} + \sigma^2},
\end{equation}
where $\sigma^2$ represents the variance of additive white Gaussian noise (AWGN). This corresponds to an orthogonal frequency division multiple access (OFDMA) system with full frequency reuse between cells. The bandwidth efficiency of BS $k$ serving DP $m$ is denoted as $\eta_{km}^{BW}$ \cite{majewski10}, and the total available system bandwidth is $W$. Gains in data rate achievable in multi-antenna systems can also be accounted for through the bandwidth efficiency parameter $\eta_{km}^{BW}$. The achieved radio downlink bandwidth \cite{morgensen07} of cell $k$ in DP $m$ can, e.g., be determined as
\begin{equation}
B_{km} = \eta_{km}^{BW} W \log_2 \left( 1 + \gamma_{km} \right)
\end{equation}
where the interferences are treated as noise. To satisfy the data demand $D_m$ in DP $m$, cell $k$ needs to allocate the fraction $D_m/B_{km}$ of its resources. In order to model the allocation of DPs to BSs, we use the binary parameter
\begin{equation}
A_{km} = \begin{cases}
1 \; \text{if DP} \; m \; \text{is allocated to cell} \; k \\
0 \; \text{otherwise}
\end{cases},
\end{equation}
and we denote as $\bs{A} \in \{0,1\}^{K \times M}$ the combined allocation matrix. We assume that each DP is allocated to a single BS, so that $\sum_{k=1}^K A_{km}=1$. To determine the fraction of available resources required by a cell in order to satisfy the data demand $D_m$ of all its allocated DPs $m$, as specified by $A_{km}$, we compute the total load factor $\rho_k$ of cell $k$ \cite{majewski10,siomina12a} as
\begin{equation} \label{eq:load}
\rho_k = \sum_{m=1}^M A_{km} \frac{D_m}{B_{km}} = \sum_{m=1}^M A_{km} \frac{D_m}{W \eta_{km}^{BW}} \frac{1}{\log_2 \left( 1 + \gamma_{km} \right)}.
\end{equation}
 We further define the vector of load factors $\bs{\rho} = \left[ \rho_1, \rho_2, \ldots, \rho_K \right]^T$. It is observable that $\rho_k>0$ holds and that cell $k$ is not overloaded if $\rho_k \leq 1$. An overloaded cell, with $\rho_k > 1$, cannot serve the minimum data demands $D_m$ of its allocated DPs under the present SINR conditions. Under these circumstances, new connections cannot be established, and existing connections have to be dropped. Let
\begin{equation} \label{eq:f}
f(\gamma) = \frac{1}{\log_2\left( 1 + \gamma \right)}
\end{equation}
denote, for later reference, the inverse rate, measured in time per transmitted bit, corresponding to a link with SINR $\gamma$.\\ 
The interference term $\sum_{j=1,\ldots,K \backslash \{k\}} p_j g_{jm} + \sigma^2$ in Eq.~\eqref{eq:sinr} can be weighted with the load factors $\rho_j$ of interfering cells, to account for the fact that lightly loaded cells do not need to fully use their available time-frequency resources and therefore generate on average lower levels of interference than heavily loaded ones \cite{majewski10,siomina12a,yuan12}. Because even lightly loaded cells might fully interfere with each other if there are no coordination mechanisms employed, we will use, without loss of generality, the "worst-case" assumption of full interference between active cells in our simulations. \\
To indicate the on-off activity status of cells, we introduce the binary model parameter
\begin{equation}
x_k = \begin{cases}
1 \quad \text{if cell} \; k \; \text{is active}\\
0 \quad \text{otherwise}
\end{cases}.
\end{equation}
and the vector $\bs{x} = \left[x_1,x_2,\ldots,x_K\right]^T$ representing the activity status of all cells in the network. \\
We define the energy consumption of cell k as
\begin{equation} \label{eq:cellenergy}
E_k = \Gamma \left( x_k,p_k,\rho_k \right)
\end{equation}
where $\Gamma \left( x_k,p_k,\rho_k \right)$ is an arbitrary linearly increasing function of the cell's on-off status $x_k$, transmit power $p_k$ and load $\rho_k$. For example, the energy consumption function used in Eq.~\eqref{eq:cellenergy} can be defined as
\begin{equation} \label{eq:speceng}
\Gamma  \left( x_k,\tilde{p}_k,\rho_k \right) = T_0 P_k^{\mathrm{MAX}} \left(\kappa_1  x_k + \kappa_2 \frac{\tilde{p}_k}{P_k^{\mathrm{MAX}}} + \kappa_3\tilde{\rho}_k\right)
\end{equation}
where the parameters $\kappa_1$, $\kappa_2$ and $\kappa_3$ are weighting factors for the cell's energy consumption based on the on-off status, transmit power, and load factor, respectively, and $T_0$ is a time constant. The load factor of a cell can impact its power consumption because it reflects the amount of its utilization \cite{son11}. Recent network models therefore have established that, especially for small cells, the power consumption is best modeled as a function of the cell load in addition to the transmit power \cite{utschick16,yuan16}. Note that the terms $x_k$, $\tilde{p}_k/P_k^{\mathrm{MAX}}$ and $\tilde{\rho}_k$ cannot exceed the value $1$, for each cell $k$. For more sophisticated models for the power consumption of mobile communication BSs, which incorporate energy consumption of wired backhaul, and individual factors for all components of the BS we refer to \cite{hasan11,bogucka11,desset12,deryuck14}. Since our model can use any combination of the three factors in Eq.~\eqref{eq:cellenergy}, we obtain a highly flexible approach for energy minimization, which is shown in the following Section \ref{sec:opt}.

\section{Energy Consumption Optimization} \label{sec:opt}
In this section, we formulate an optimization problem to minimize the energy consumption of the wireless network as defined in Eq.~\eqref{eq:cellenergy}, subject to DP to BS allocation-, minimum SINR- and cell load constraints. Physical layer requirements of the wireless communications standard, such as the used modulation and coding scheme in LTE-A, impose a minimum SINR level $\gamma^\text{MIN}>0$ for the transmission link to provide a non-zero rate, and define $\tau^\text{MAX} = f \left(\gamma^\text{MIN} \right)$. The load term $(D_m \tau^\text{MAX})/(W \eta_{km}^{BW})$ therefore is the highest load DP $m$ contributes to the overall load \eqref{eq:load} of cell $k$.\\
The SINR threshold $\gamma^\mathrm{MAX}$ denotes the SINR-level where, for $\gamma \geq \gamma^\mathrm{MAX}$, the highest available modulation- and coding scheme is used such that the maximum rate is achieved, and a further increase in SINR is not associated with an additional increase in rate. The inverse of the log-term in Eq.~\eqref{eq:load} for SINR-levels $\gamma \geq \gamma^\mathrm{MAX}$ is denoted as $\tau^\text{MIN} = f\left(\gamma^\text{MAX} \right)$. We further define
\begin{equation} \label{eq:fplus}
f_{\tau^\text{MIN}}^+(\gamma) = \max \left\{ f(\gamma),\tau^\text{MIN} \right\}.
\end{equation}
For the allocation of DPs to cells, we assume that cell range expansion is being utilized \cite{yuan12,ye13}, with $\theta_k$ denoting the bias value of cell $k$. DP $m$ is allocated to the cell $k$ that provides the highest product of received signal power $p_k g_km$ and bias value $\theta_k$.
Using as optimization parameters the binary cell activity indicator $\bs{x} \in \{0,1\}^{K \times 1}$ and allocation indicator $\bs{A} \in \{0,1\}^{M \times K}$, the continuous transmit power parameter $\bs{p} \in \mathbb{R}_{0+}^{K \times 1}$ and the cell load $\bs{\rho} \in \mathbb{R}_{0+}^{K \times 1}$, the energy minimization problem can be formulated as following:
\begin{subequations} \label{eq:orig1}
	\begin{align}
	\underset{\boldsymbol{x},\boldsymbol{p},\boldsymbol{A},\bs{\rho}}{\mathrm{minimize}} \qquad & \sum_{k=1}^K \Gamma \left( x_k,p_k,\rho_k \right) & \; & \label{eq:orig1_obj} \\ 
	\mathrm{subject \; to} \qquad	& P_k^{\mathrm{MIN}} \leq p_k \leq P_k^{\mathrm{MAX}} && \forall k \label{eq:orig1_pconst}\\
	& \sum_{k=1}^K A_{km} = 1 && \forall m \label{eq:orig1_alloc}\\
	& \sum_{k=1}^K A_{km} \leq x_k && \forall k,m \label{eq:orig1_ax}\\
	& \sum_{k=1}^K A_{km} \theta_k p_k g_{km} \geq x_j \theta_j p_j g_{jm} && \forall j,m \label{eq:orig1_mpow}\\
	& \sum_{k=1}^K A_{km} p_k g_{km} - \gamma^\mathrm{MIN} \left( \sum_j x_j (1-A_{jm})p_j g_{jm} + \sigma^2 \right) \geq 0 && \forall m \label{eq:orig1_sinrcon}\\
	& \rho_k = \sum_{m=1}^M A_{km} \frac{D_m}{W \eta_{km}^\mathrm{BW}} f_{\tau^\text{MIN}}^+ \left( \frac{p_k g_{km}}{\sum_{j=1}^K x_j (1-A_{jm})p_j g_{jm} + \sigma^2} \right)  && \forall k \label{eq:orig1_loadc} \\
	& \rho_k \leq 1 && \forall k \label{eq:orig1_load} \\
	& x_k,A_{km} \in \{0,1\} && \forall k,m \\
	& p_k \in \mathbb{R}_{0+} && \forall k
	\end{align}
\end{subequations}
In problem \eqref{eq:orig1}, the objective \eqref{eq:orig1_obj} aims to minimize the overall systems' energy consumption, which is the sum of the energy consumption of individual cells as defined in \eqref{eq:cellenergy} and \eqref{eq:speceng}. The constraint \eqref{eq:orig1_pconst} defines the feasible transmit power range of cell $k$ restricted according to \eqref{eq:pconst}. Each DP $m$ is served by exactly one cell $k$, and only active cells $\{k|x_k=1\}$ can serve any DP, as specified by \eqref{eq:orig1_alloc} and \eqref{eq:orig1_ax}, respectively. Constraint \eqref{eq:orig1_mpow} enforces that, each DP $m$ is allocated to the cell $k$ that provides highest product of received signal power and bias value \footnote{Typically the DP is allocated to the cell providing the highest received signal power, but this leads to an underutilization of the low-power small cells. If so-called "range expansion" is utilized, the signal power from small cells is weighted with a bias value, which corresponds to an increased coverage area \cite{ye13}.}. The load constraint that cell $k$ has to satisfy, as defined in \eqref{eq:load}, is specified in \eqref{eq:orig1_load}. \\
Problem \eqref{eq:orig1} is a combinatorial and nonconvex MINLP, and thus generally very difficult to solve. While significant advancements have been made for convex MINLPs \cite{hijazi2012mixed,belotti13}, it is universally agreed upon that nonconvex MINLPs pose a significant computational challenge where the chances of finding an optimal solution to any given problem highly depend on the problem size and structure \cite{floudas89,trespalacios2014review}. To maintain robustness and scalability for schemes based on network optimization problems, it is therefore advisable to find an MILP that represents a linear inner approximation or a linear reformulation of the original MINLP. The objective function \eqref{eq:orig1_obj} and constraints \eqref{eq:orig1_mpow}, \eqref{eq:orig1_sinrcon} and \eqref{eq:orig1_load} contain the bilinear term $x_k p_k$. We introduce a new variable $\tilde{p}_k \triangleq p_k x_k$ and reformulate \eqref{eq:orig1} as the following equivalent problem:
\begin{subequations} \label{eq:ptilde}
	\begin{align}
	\underset{\boldsymbol{x},\boldsymbol{\tilde{p}},\boldsymbol{A},\bs{\rho}}{\mathrm{minimize}} \qquad & \sum_{k=1}^K \Gamma \left( x_k,\tilde{p}_k,\rho_k \right) & \; & \\ 
	\mathrm{subject \; to} \qquad & x_k P_k^{\mathrm{MIN}} \leq \tilde{p}_k \leq x_k P_k^{\mathrm{MAX}} && \forall k \label{eq:ptilde_pconst} \\
	& \eqref{eq:orig1_alloc} - \eqref{eq:orig1_ax} \nonumber\\
	& \sum_{k=1}^K A_{km} \theta_k \tilde{p}_k g_{km} \geq \theta_j \tilde{p}_j g_{jm} && \forall j,m \label{eq:ptilde_mpow}\\
	& \sum_{k=1}^K A_{km} \tilde{p}_k g_{km} - \gamma^\mathrm{MIN} \left( \sum_j (1-A_{jm}) \tilde{p}_j g_{jm} + \sigma^2 \right) \geq 0 && \forall m \label{eq:ptilde_sinrcon}\\
	& \rho_k = \sum_{m=1}^M A_{km} \frac{D_m}{W \eta_{km}^\mathrm{BW}} f_{\tau^\text{MIN}}^+ \left( \frac{\tilde{p}_k g_{km}}{\sum_{j=1}^K (1-A_{jm}) \tilde{p}_j g_{jm} + \sigma^2} \right)  && \forall k \label{eq:ptilde_loadc} \\
	& \rho_k \leq 1 && \forall k \label{eq:ptilde_load} \\
	& x_k,A_{km} \in \{0,1\} && \forall k,m \\
	& \tilde{p}_k \in \mathbb{R}_{0+} && \forall k
	\end{align}
\end{subequations}
Using a lifting strategy, we will in the following introduce auxiliary parameters to represent bilinear products of optimization variables, which a more tractable, linear problem structure at the cost of increased problem dimensionality. Towards this aim, the bilinear products of binary allocation parameters $A_{km}$ and cell transmit powers $\tilde{p}_k$ in Eqs.~\eqref{eq:ptilde_mpow}, \eqref{eq:ptilde_sinrcon} and \eqref{eq:ptilde_loadc} have to be linearized. 
We define the set
\begin{equation} \label{eq:linset}
\mathcal{L} \coloneqq \{ \left(r,\overline{r},b,a \right) \in \mathbb{R}_{0+} \times \mathbb{R}^+ \times \{0,1\} \times \mathbb{R}_{0+}:  a\geq r-(1-b)\overline{r}, a\leq r,a\leq b \overline{r} \}
\end{equation}
with the binary parameter $b$ and the real parameter $r$ with $0 \leq r \leq \overline{r}$. The inequalities defining $\mathcal{L}$ in $\eqref{eq:linset}$ are affine in $r$, $b$ and $a$, and $\left(r,\overline{r},b,a \right) \in \mathcal{L}$ enforces $a=rb$, which will be used in the following reformulations to linearize bilinear products of binary and continuous optimization parameters \cite{Liberti2006}.\\
We introduce an new variable $\Omega_{km}$ and the corresponding matrix $\bs{\Omega} \in \mathbb{R}_{0+}^{K \times M}$. For the proposed lifting approach, we install $\left(\tilde{p}_k, P_k^{\mathrm{MAX}}, A_{km},\Omega_{km} \right) \in \mathcal{L} \; \forall k,m$ in problem \eqref{eq:ptilde}, which enforces that $\Omega_{km} = \tilde{p}_k A_{km}$, such that we can reformulate \eqref{eq:orig1} as:
\begin{subequations} \label{eq:orig2}
	\begin{align}
	\underset{\boldsymbol{x},\boldsymbol{p},\boldsymbol{A},\bs{\rho},\boldsymbol{\Omega}}{\mathrm{minimize}} \qquad & \sum_{k=1}^K \Gamma \left( x_k,\tilde{p}_k,\rho_k \right) & \; & \label{eq:orig2_obj} \\ 
	\mathrm{subject \; to} \qquad	& \eqref{eq:orig1_alloc} - \eqref{eq:orig1_ax}, \eqref{eq:ptilde_pconst}, \eqref{eq:linset} \nonumber \\
	& \sum_{k=1}^K \Omega_{km} \theta_k g_{km} \geq x_j \theta_j \tilde{p}_j  g_{jm} && \forall j,m \label{eq:orig2_mpow}\\
	& \sum_{k=1}^K \Omega_{km} g_{km} - \gamma^\mathrm{MIN} \left( \sum_j (1-\Omega_{jm}) g_{jm} + \sigma^2 \right) \geq 0 && \forall m \label{eq:orig2_sinrcon}\\
	& \rho_k = \sum_{m=1}^M A_{km} \frac{D_m}{W \eta_{km}^\mathrm{BW}} f_{\tau^\text{MIN}}^+ \left(\frac{\tilde{p}_k g_{km}}{\sum_{j=1,\ldots,K} (\tilde{p}_j-\Omega_{jm}) g_{jm} + \sigma^2} \right) && \forall k \label{eq:orig2_loadc} \\
	& \rho_k \leq 1 && \forall k \label{eq:orig2_load} \\
	& \left(\tilde{p}_k, P_k^{\mathrm{MAX}}, A_{km},\Omega_{km} \right) \in \mathcal{L} && \forall k,m \label{eq:powcs_omega} \\
	& x_k,A_{km} \in \{0,1\} && \forall k,m \\
	& \tilde{p}_k,\Omega_{km} \in \mathbb{R}_{0+} && \forall k
	\end{align}
\end{subequations}
From \eqref{eq:orig1} to \eqref{eq:orig2}, the auxiliary parameter $\bs{\Omega}$ has been used in constraints \eqref{eq:orig2_mpow}, \eqref{eq:orig2_sinrcon} and \eqref{eq:orig2_loadc} to replace $\Omega_{km} = \tilde{p}_k A_{km}$, whereas the remaining optimization parameters remain unchanged. The solution of problem \eqref{eq:orig2} can therefore be used to easily obtain the corresponding solutions for problem \eqref{eq:orig1} and vice-versa. Thus, both formulations can be considered equivalent.\\	
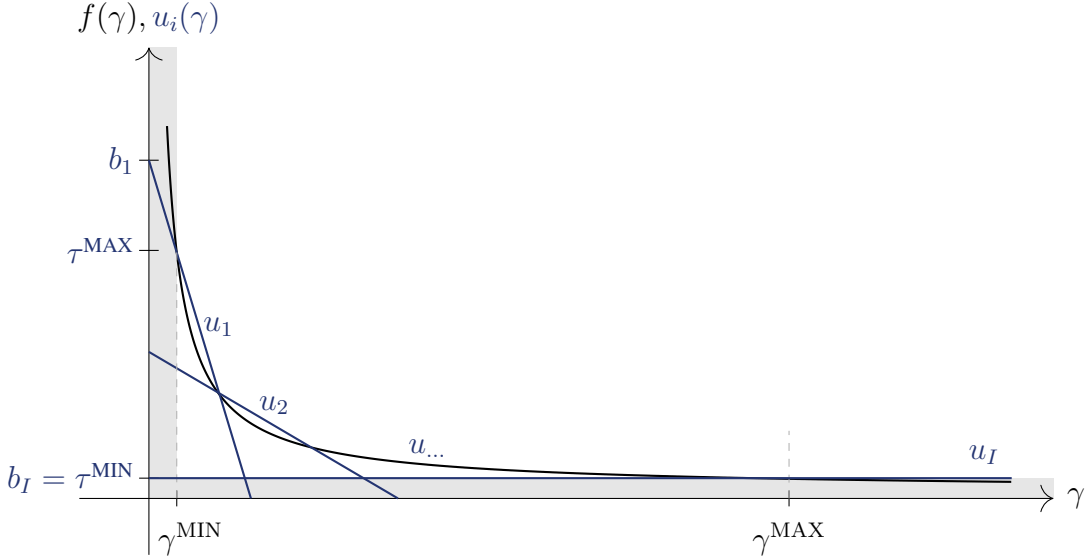
\begin{figure}
	\centering
	\begin{tikzpicture} [xscale=3.7,yscale=3]
	\definecolor{tud1d}{RGB}{36,53,114}
	\definecolor{tud4d}{RGB}{106,139,55}
	\definecolor{tud8d}{RGB}{40,80,100}
	
	\draw[-{>[scale=2]}] (-0.25,0) -- (3.25,0) node[right=1pt] {$\gamma$};
	\draw[-{>[scale=2]}] (0,-0.25) -- (0,2) node[above=1pt] {$f(\gamma), \color{tud1d} u_i(\gamma)$};
	
	\draw[smooth,samples=1001,domain=0.065:3.1,thick] plot (\x,{0.15/log2(1+\x)});
	

	\draw (-1pt,1.5) -- (1pt,1.5) node[left=5pt] {\color{tud1d} $b_1$};
	\draw (-1pt,1.1) -- (1pt,1.1) node[left=5pt] {\color{tud1d} $\tau^\text{MAX}$};
	\draw (-1pt,0.09) -- (1pt,0.09) node[left=5pt] {\color{tud1d} $b_{I}=\tau^\text{MIN}$};
	\draw[color=tud1d,thick] (0,1.5) -- (0.3666,0);
	\draw[color=tud1d,thick] (0,0.65) -- (0.895,0);
	
	\draw (0.1,-1pt) -- (0.1,1pt);
	\draw (0.15,1pt) node[left,below=8pt] {$\gamma^\mathrm{MIN}$};
	\draw (2.3,-1pt) -- (2.3,1pt);
	\draw (2.3,1pt) node[left,below=8pt] {$\gamma^\mathrm{MAX}$};
	
	\draw[dashed,color=black!30] (0.1,1.0909) -- (0.1,0);
	\draw (0.25,0.65) node[above=2pt] {\color{tud1d} $u_1$};
	\draw[dashed,color=black!30] (2.3,0) -- (2.3,0.3);
	\draw (0.45,0.3) node[above=2pt] {\color{tud1d} $u_2$};
	\draw (1,0.11) node[above=2pt] {\color{tud1d} $u_{\ldots}$};
	\draw [thick,color=tud1d] (0,0.09) -- (3.1,0.09);
	\draw (3,0.08) node[above=2pt] {\color{tud1d} $u_{I}$};
	
	\begin{pgfonlayer}{background}
	\fill[fill=black!10] (0,0) rectangle (0.1,2);
	\fill[fill=black!10] (0,0) rectangle (3.25,0.09);
	\end{pgfonlayer}
	\end{tikzpicture}
	\caption{Illustration of the piecewise linear over-approximation of the cell load function $f(\gamma)$ with the linear functions $u_i(\gamma)$ in the SINR interval $\gamma^\mathrm{MIN} \leq \gamma \leq \gamma^\mathrm{MAX}$.}
	\label{fig:linapprox}
\end{figure}
Problem \eqref{eq:orig2} is an integer linear program except for constraint \eqref{eq:orig2_loadc}, which is nonlinear due to the log-term in the function $f_{\tau^\text{MIN}}^+(\gamma)$ as defined in \eqref{eq:fplus}, the fractional SINR-term and the allocation factor $A_{km}$. In the following, we propose an affine inner approximation of \eqref{eq:orig2_loadc}-\eqref{eq:orig2_load}. We define a set of $I$ linear functions 
\begin{equation} \label{eq:u}
u_i(\gamma) = \alpha_i \gamma + \beta_i, i=1,\ldots,I,
\end{equation}
which satisfy the upper bound property
\begin{equation} \label{eq:ufkt}
\underset{i}{\max} \; u_i(\gamma) \geq f_{\tau^\text{MIN}}^+(\gamma) \quad \forall \; \gamma \geq \gamma^\mathrm{MIN},
\end{equation}
as illustrated in Fig.~\ref{fig:linapprox}. Since $f(\gamma)$ in \eqref{eq:f} is strictly decreasing, we conclude that all $u_i(\gamma)$ can be designed such that $\alpha_i \leq 0 \ \forall \ i$. To approximate the load for $\gamma \geq \gamma^\mathrm{MAX}$, as depicted in Fig.~\ref{fig:linapprox}, a constant function can be used with $u_I(\gamma) = \beta_I = \tau^\text{MIN}$. The issue of designing a suitable set of $u_i$ that keep the maximum absolute approximation error below a selectable threshold $\epsilon$ is discussed in Appendix \ref{sec:applin}. \\
We introduce the optimization parameter $\mu_{km}$ designed to be an upper bound of the load term in Eq.~\eqref{eq:orig2_loadc}, such that
\begin{equation} \label{eq:mufirst}
\mu_{km} \geq  u_i(\gamma) \quad \forall \; i,\gamma \geq \gamma^\mathrm{MIN}
\end{equation}
and the corresponding matrix $\bs{\mu} \in  \mathbb{R}_{0+}^{K \times M}$. For the interval $\gamma^\mathrm{MIN} \leq \gamma \leq \gamma^\mathrm{MAX}$, we reformulate the log-term contained in the  function $f_{\tau^\text{MIN}}^+(\gamma)$ in the constraint \eqref{eq:orig2_loadc} as
\begin{equation} \label{eq:orig2_loadc_ref1}
\rho_k = \sum_{m=1}^M A_{km} \frac{D_m}{W \eta_{km}^\mathrm{BW}} \mu_{km}
\end{equation}
where for \eqref{eq:u}-\eqref{eq:mufirst}
\begin{equation}\label{eq:orig2_loadc_ref2}
\mu_{km} \geq \alpha_i \frac{\tilde{p}_k g_{km}}{\sum_{j=1,\ldots,K} (1-\Omega_{jm}) g_{jm} + \sigma^2} + \beta_i \quad \forall \; i,k,m
\end{equation}
 We further denote the product of $\mu_{km}$ and allocation parameter $A_{km}$ as $\Lambda_{km}=\mu_{km} A_{km}$ and the corresponding matrix as $\bs{\Lambda} \in \mathbb{R}_{0+}^{K \times M}$. This bilinear product formulation for $\bs{\Lambda}$ is replaced by a linear reformulation using \eqref{eq:linset} by adding the constraint that $\left(\mu_{km}, \beta_1, A_{km},\Lambda_{km} \right) \in \mathcal{L}$. \\
In order to approximate the interference levels in the denominator of the SINR term Eq.~\eqref{eq:sinr}, we introduce the scalar interference levels $\Psi_{nkm}$ with interference scenario index $n=1,\ldots,N$, and the corresponding three-dimensional scalar tensor $\boldsymbol{\Psi} \in \mathbb{R}_{0+}^{N \times K \times M}$. We also introduce a binary interference scenario selection parameter $\phi_{nkm} $ and the corresponding three-dimensional binary tensor $\boldsymbol{\phi} \in \{0,1\}^{N \times K \times M}$. To ensure that the solution of the approximate problem is always feasible for the original, we add the constraint that the selected discrete interference level is always an over-approximation of the actual interference:
\begin{equation} \label{eq:powcs_intcons}
\sum_{n=1}^N \phi_{nkm} \Psi_{nkm} \geq \sum_{j=1,\ldots,K} (1-\Omega_{jm}) g_{jm} + \sigma^2 \; \forall k,m \\ 
\end{equation}
When implementing the selection parameter $\boldsymbol{\phi}$ in Eq.~\eqref{eq:orig2_loadc_ref2}, we replace the bilinear product $\tilde{p}_k g_{km} \phi_{nkm}$ with an auxiliary parameter, for which we introduce the lifting variable $\Phi_{nkm} = \tilde{p}_k g_{km} \phi_{nkm}$ with the corresponding tensor variable $\boldsymbol{\Phi} \in \mathbb{R}_{0+}^{N \times K \times M}$. Again, the product computation of $\boldsymbol{\Phi}$ will be replaced by an auxiliary parameter using \eqref{eq:linset} by adding the constraint $\left( \tilde{p}_k g_{km}, P_k^{\mathrm{MAX}} g_{km}, \phi_{nkm},\Phi_{nkm} \right) \in \mathcal{L} \; \forall n,k,m$. \\
The proposed linear inner approximation of \eqref{eq:orig2} is the following:
\begin{subequations} \label{eq:powcs}
	\begin{align}
	\underset{\boldsymbol{x},\boldsymbol{\tilde{p}},\boldsymbol{A},\bs{\tilde{\rho}},\boldsymbol{\Omega},\boldsymbol{\mu},\boldsymbol{\Lambda},\boldsymbol{\phi},\boldsymbol{\Phi}}{\mathrm{minimize}} \qquad & \sum_{k=1}^K \Gamma \left( x_k,\tilde{p}_k,\tilde{\rho}_k \right) & \; & \label{eq:powcs_obj}\\ 
	\mathrm{subject \; to} \qquad & \eqref{eq:orig1_alloc} - \eqref{eq:orig1_ax}, \eqref{eq:ptilde_pconst}, \eqref{eq:linset}, \eqref{eq:orig2_mpow}-\eqref{eq:orig2_sinrcon}, \eqref{eq:powcs_omega}, \eqref{eq:powcs_intcons} \nonumber \\
	& \tilde{\rho}_k = \sum_{m=1}^M \left( \frac{d_m}{W \eta^{\mathrm{BW}}} \Lambda_{km} \right) && \forall k \label{eq:powcs_loadc} \\
	& \tilde{\rho_k} \leq 1 && \forall k \label{eq:powcs_load} \\
	& \sum_{n=1}^N \phi_{nkm} = 1 && \forall k,m \label{eq:powcs_intsum} \\
	& \mu_{km} \geq \alpha_i \sum_{n=1}^N \frac{\Phi_{nkm}}{\Psi_{nkm}} + \beta_i && \forall i,k,m \label{eq:powcs_lineq} \\
	& \left(\mu_{km}, \beta_1, A_{km},\Lambda_{km} \right) \in \mathcal{L} && \forall k,m \label{eq:powcs_lambda} \\
	& \left( \tilde{p}_k g_{km}, P_k^{\mathrm{MAX}} g_{km}, \phi_{nkm},\Phi_{nkm} \right) \in \mathcal{L} && \forall n,k,m  \label{eq:powcs_phi} \\
	& x_k,A_{km},\phi_{nkm} \in \{0,1\} && \forall n,k,m \\
	& \tilde{p}_k, \Omega_{km},\mu_{km},\Lambda_{km},\Phi_{nkm} \in \mathbb{R}_{0+} && \forall n,k,m
	\end{align}
\end{subequations}
\begin{proposition} \label{th:consproof}
	Problem \eqref{eq:powcs} is an inner approximation of problem \eqref{eq:orig2}, i.e. for every point $\{\bs{x},\bs{p},\bs{A}\}$ solving \eqref{eq:powcs} a feasible point of \eqref{eq:orig2} can be constructed.
\end{proposition}	
\begin{proof}
	The transmit power constraints \eqref{eq:ptilde_pconst}, the allocation constraints \eqref{eq:orig1_alloc}-\eqref{eq:orig1_ax} and the signal power constraints \eqref{eq:orig2_mpow}-\eqref{eq:orig2_sinrcon} are identical in problem \eqref{eq:orig2} and \eqref{eq:powcs}.\\
	The proposition therefore holds if the load in \eqref{eq:powcs_loadc} is an inner approximation of that in \eqref{eq:orig2_loadc}, specifically if
	\begin{equation} \label{eq:proof_1}
	\sum_{m=1}^M \left( \frac{D_m}{W \eta^{\mathrm{BW}}} \Lambda_{km} \right) \geq \sum_{m=1}^M A_{km} \frac{D_m}{W \eta_{km}^\mathrm{BW}} \frac{1}{\log_2 \left(1 + \frac{\tilde{p}_k g_{km}}{\sum_{j=1,\ldots,K} (1-\Omega_{jm}) g_{jm} + \sigma^2} \right)} \quad \forall \; k.
	\end{equation}
	Due to \eqref{eq:powcs_lambda}, we have $\Lambda_{km} = \mu_{km} A_{km}$, therefore \eqref{eq:proof_1} is satisfied if
	\begin{equation} \label{eq:proof2}
	\mu_{km} \geq \frac{1}{\log_2 \left(1 + \frac{\tilde{p}_k g_{km}}{\sum_{j=1,\ldots,K} (1-\Omega_{jm}) g_{jm} + \sigma^2} \right)} \quad \forall \; k,m,
	\end{equation}
	from which, with \eqref{eq:powcs_lineq} and \eqref{eq:ufkt} applied to the left- and right-hand side of Eq.~\eqref{eq:proof2}, respectively, we obtain
	\begin{equation} \label{eq:proof3}
	\alpha_i \sum_{n=1}^N \frac{\Phi_{nkm}}{\Psi_{nkm}} + \beta_i \geq \alpha_i \sum_{n=1}^N \frac{\tilde{p}_k g_{km}}{\sum_{j=1,\ldots,K} (1-\Omega_{jm}) g_{jm} + \sigma^2} + \beta_i \quad \forall \; i,k,m.
	\end{equation}
	Due to the constraints \eqref{eq:powcs_phi}, which implement the bilinear constraint $\Phi_{nkm} = \tilde{p}_k g_{km} \phi_{nkm}$, and due to $\phi_{nkm} \in \{0,1\} \ \forall \ n,k,m$, we have
	\begin{equation} \label{eq:proof_scenfracs}
	\sum_{n=1}^N \frac{\Phi_{nkm}}{\Psi_{nkm}} = \frac{\tilde{p}_k g_{km}}{\sum_{n=1}^N \phi_{nkm} \Psi_{nkm}} \quad \forall \; n,k,m.
	\end{equation}
	Substituting \eqref{eq:proof_scenfracs} in the left-hand side of \eqref{eq:proof3}, we obtain the inequality
	\begin{equation}
	\alpha_i \sum_{n=1}^N \frac{\tilde{p}_k g_{km}}{\sum_{n=1}^N \Psi_{nkm}} + \beta_i \geq \alpha_i \sum_{n=1}^N \frac{\tilde{p}_k g_{km}}{\sum_{j=1,\ldots,K} (1-\Omega_{jm}) g_{jm} + \sigma^2} + \beta_i \quad \forall \; i,k,m,
	\end{equation}
	which holds due to the constraint \eqref{eq:powcs_intcons} for $\alpha_i \leq 0 \ \forall \ i$, thus proving the proposition.
\end{proof}
The tightness of the approximating problem \eqref{eq:powcs} with regards to problem \eqref{eq:orig2} depends on two factors. The first factor is related to how closely the linear functions $u_i$ approximate the load function as in Eq.~\eqref{eq:ufkt}. The second factor is how closely the discrete interference levels $\Psi_{nkm}$ approximate the actual interference level $\sum_{j=1,\ldots,K} (1-\Omega_{jm}) g_{jm} + \sigma^2$. Proposition \ref{th:consproof} holds irrespectively the choice of the discrete interference levels $\Psi_{nkm}$. Certain changes in interference levels, specifically the removal of strongest interferers, cause large differences in the load caused by a DP. The levels $\Psi_{nkm}$ can be chosen in such a way that these changes can be reflected by the selection of a different interference scenario. The accuracy of the inner approximation can be improved by using a larger number of interference levels, at the cost of increased problem complexity. \\
We propose to construct, for each pair $(m,k)$ of DP $m$ allocated to cell $k$, interference levels $\Psi_{nkm}$ that mainly reflect transmit power changes of the first- and second-strongest interferers \cite{maattanen12,ramos17,gulati15}. With
\begin{equation}
v = \argmax_{j \setminus \{k\}} (p_{j}g_{jm})
\end{equation}
and
\begin{equation}
w = \argmax_{j \setminus \{k,v\} \}} (p_{j}g_{jm})
\end{equation}
we compute our interference levels as
\begin{equation} \label{eq:intlevels}
\Psi_{nkm} = l_{n}^{\mathrm{P}} p_{v}g_{vm} + l_{n}^{\mathrm{S}} p_{w}g_{wm} + l_{n}^{\mathrm{R}} \sum_{j \setminus \{k,v,w\}} p_{j}g_{jm} + \sigma^2,
\end{equation}
where the parameters $l_{n}^{\mathrm{P}}$,$l_{n}^{\mathrm{S}}$ and $l_{n}^{\mathrm{R}}$ denote the weighting factors for primary-, secondary- and remaining interferers, respectively. Keeping in mind that we focus on transmit power changes for the first- and second strongest interferers, a suitable set of weighting factors to compute the interference levels $\Psi_{nkm}$ is shown, for example, in Table \ref{tb:intweights}.
\begin{table}
	\setlength\extrarowheight{2pt}
	\begin{center}
	\caption{Weighting factors for computation of interference scenarios $\Psi_{nkm}$, used for an over-approximation of the actual interference level.}
	\label{tb:intweights}
		\begin{tabular}{| l || c | c | c | c | c | c | c | }
			\hline
			$\qquad n=$ & $1$ & $2$ & $3$ & $4$ & $5$ & $6$ & $7$  \\ \hline \hline
			$l_{n}^{\mathrm{P}}$ & $1$ &  $0.75$ & $0.5$ & $0.25$ & $0$ & $0$ & $0$  \\ \hline
			$l_{n}^{\mathrm{S}}$ & $1$ &  $1$ & $1$ & $1$ & $1$ & $0$ & $0$  \\ \hline
			$l_{n}^{\mathrm{R}}$ & $1$ &  $1$ & $1$ & $1$ & $1$ & $1$ & $0$  \\ \hline
		\end{tabular}
	\end{center}
\end{table}

\section{Simulation Results} \label{sec:simres}

\begin{figure}
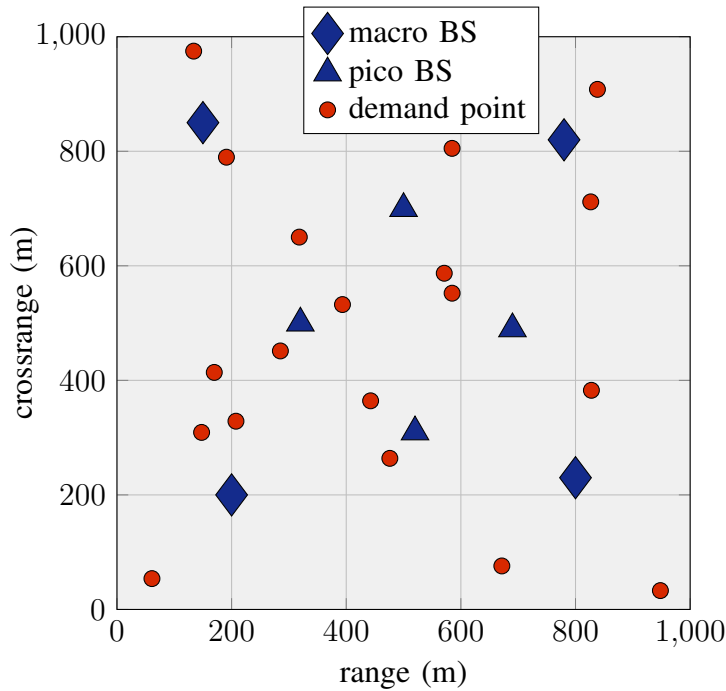

	\centering	
	\include{fF_netset}
	\caption{Illustration of the network scenario with 4 macro- and 4 small cells and an example distribution of 20 DPs. The network area is 1000m times 1000m and path loss between cells and DPs is modeled according to 3GPP TS 36.814 specification.}
	\label{fig:netset1}
\end{figure}

\begin{table} 
	\setlength\extrarowheight{2pt}
	\begin{center}
		\caption{Simulation parameters of a downlink LTE network. The transmit power of the cells is optimized inside a $10 \mathrm{dB}$ interval. Results are averaged over 5000 simulations with fixed base station positions and randomly distributed DPs.}
		\label{tb:LTE}
		\begin{tabular}{| l | r |}
			\hline
			Area size & $1000 \times 1000$ m \\ \hline
			Noise power & -145 dBm/Hz \\ \hline
			System bandwidth $W$  & $20$ MHz \\ \hline
			Position of macro BS & MBS1 at [200m, 200m] \\
			& MBS2 at [150m, 850m] \\
			& MBS3 at [800m, 230m] \\
			& MBS4 at [780m, 820m] \\ \hline
			MBS transmit power range $P^\text{MIN} \ldots P^\text{MAX}$ & $36 \mathrm{dBm} \ldots 46 \mathrm{dBm}$ \\ \hline
			MBS antenna gain $\tilde{g}^\text{BS}$ & $15 \mathrm{dB}$ \\ \hline
			MBS bias value $\theta_k$ & $0 \mathrm{dB}$ \\ \hline
			Position of pico BS & PBS1 at [500m, 700m] \\
			& PBS2 at [520m, 310m] \\
			& PBS3 at [320m, 500m] \\
			& PBS4 at [690m, 490m] \\ \hline
			PBS transmit power range $P^\text{MIN} \ldots P^\text{MAX}$ & $26 \mathrm{dBm} \ldots 36 \mathrm{dBm}$ \\ \hline
			PBS antenna gain $\tilde{g}^\text{BS}$ & $5 \mathrm{dB}$ \\ \hline
			PBS bias value $\theta_k$ & $3 \mathrm{dB}$ \\ \hline
			DP antenna gain $\tilde{g}^\text{DP}$ & $0 \mathrm{dB}$ \\ \hline
			Propagation loss $\tilde{g}^\mathrm{PATH}$ & 3GPP TS 36.814 \cite{3GPP} \\ \hline
			Bandwidth efficiency $\eta^{\mathrm{BW}}$ & 0.8 \\ \hline
			SINR requirement $\gamma^{\mathrm{MIN}}$ & $-10 \mathrm{dB}$ \\ \hline
			SINR threshold $\gamma^{\mathrm{MAX}}$ & $20 \mathrm{dB}$ \\
			\hline
		\end{tabular}
	\end{center}
\end{table}

To evaluate the performance of the proposed method, we simulate a heterogeneous wireless communication network containing 4 macro- and 4 pico cells as illustrated in Fig.~\ref{fig:netset1}. The selected system parameters are summarized in Table \ref{tb:LTE}. The selectable transmit power range and antenna gains are chosen as $36 \mathrm{dBm} - 46 \mathrm{dBm}$ with $15 \mathrm{dB}$ antenna gain for macro cells and $26 \mathrm{dBm} - 36 \mathrm{dBm}$ with $5 \mathrm{dB}$ antenna gain for small cells. A bias value of $\theta_k = 2 \mathrm{dB}$ is used for small cells to slightly increase their coverage area. The proposed method using Problem \eqref{eq:powcs} was solved using CVX for MATLAB \cite{cvx,gb08} and Gurobi as a MILP solver \cite{gurobi}. For the energy consumption modeling of cells we use Eq.~\eqref{eq:speceng} with  $\kappa_1 = 0.5$, $\kappa_2 = 0.5$ and $\kappa_3 = 0$. This implies that the power consumption of cell $k$ depends on its on-off status indicator $x_k$ and its transmit power $p_k$. The power consumption is modeled this way in order to allow comparability of the proposed MILP with an established heuristic method proposed in \cite{yuan15} that focuses on transmit power minimization. As a performance benchmark for our energy minimization algorithm we use the power scaling method introduced in \cite{yuan15}, which we extended in the following ways to make it applicable to our problem: power scaling is used for all possible configurations of all cells' on-off status $\bs{x}$. Resulting transmit powers obtained by the algorithm of \cite{yuan15} that lie below or above the bounds specified in Table \ref{tb:LTE} are projected to the lower- and upper bound respectively. Then, the best configuration that does not violate load constraints is selected as the solution. This algorithm therefore combines an exhaustive search over all configurations for $\bs{x}$ with power scaling being used in each configuration. It is in the following in all figures denoted as "power scaling + exh.~search". The second approach we use for comparison is an exhaustive search over all combinations of cells being switched on or off, with the transmit powers being fixed to $P^\mathrm{MAX}$, which we in the following indicate as "max power cell switching". The solution of the original MINLP in \eqref{eq:orig1} is unsuitable as a lower bound solution even for small problem sizes, because even for fixed binary optimization parameters the resulting continuous problem is still nonconvex. 

\begin{figure}
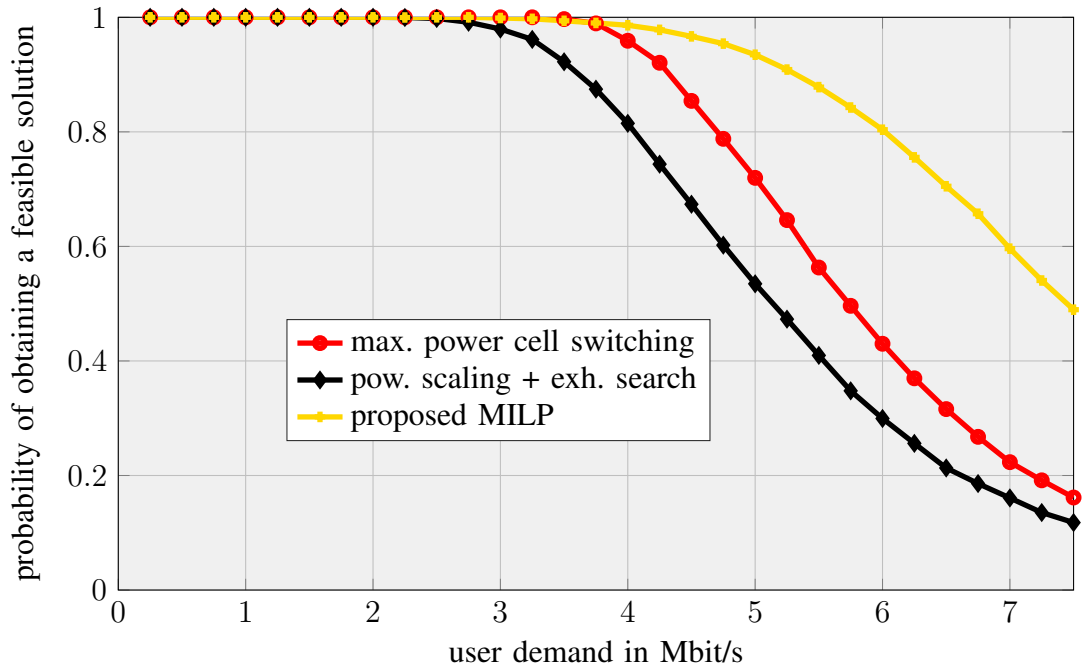

	\centering	
	\include{fP_ratiosolved}
	\caption{Probability of obtaining a feasible solution over increasing user demand, evaluated over 5000 simulations of $M=20$ randomly distributed demand points. The proposed MILP-based scheme achieves the highest solution percentage.}
	\label{fig:solvp1}
\end{figure}
Deploying $M=20$ DPs randomly in the network area illustrated in Fig.~\ref{fig:netset1}, 5000 network scenarios are generated and each DPs data demand in each scenario is scaled between $d_m = 0.25$Mbit/s and $d_m = 7.5$Mbit/s. The proposed energy-minimized solution obtained from solving problem \eqref{eq:powcs} is compared to the solutions of the aforementioned max.~power cell switching and combined power scaling and exhaustive search methods \cite{yuan15}. The probability of obtaining a feasible solution with no overloaded cells is illustrated in Fig.\ref{fig:solvp1}. The proposed MILP based method is much more likely to find a feasible and power-minimized solution even in high demand scenarios. 
In the following we discuss the performance indicators: energy consumption, cell load, and number of active cells. To ensure a fair comparison, the respective averages were computed only from those scenarios that were solved by all methods. Fig.~\ref{fig:fullcomp1} shows the average power consumption achieved by each of the three considered energy minimization schemes. The proposed MILP-based approach achieves lower power consumption levels than both the cell switching and the heuristic approach. The cell switching method noticeably achieves good performance up until about $3$Mbit/s, with the performance significantly deteriorating for higher demands.

\begin{figure}
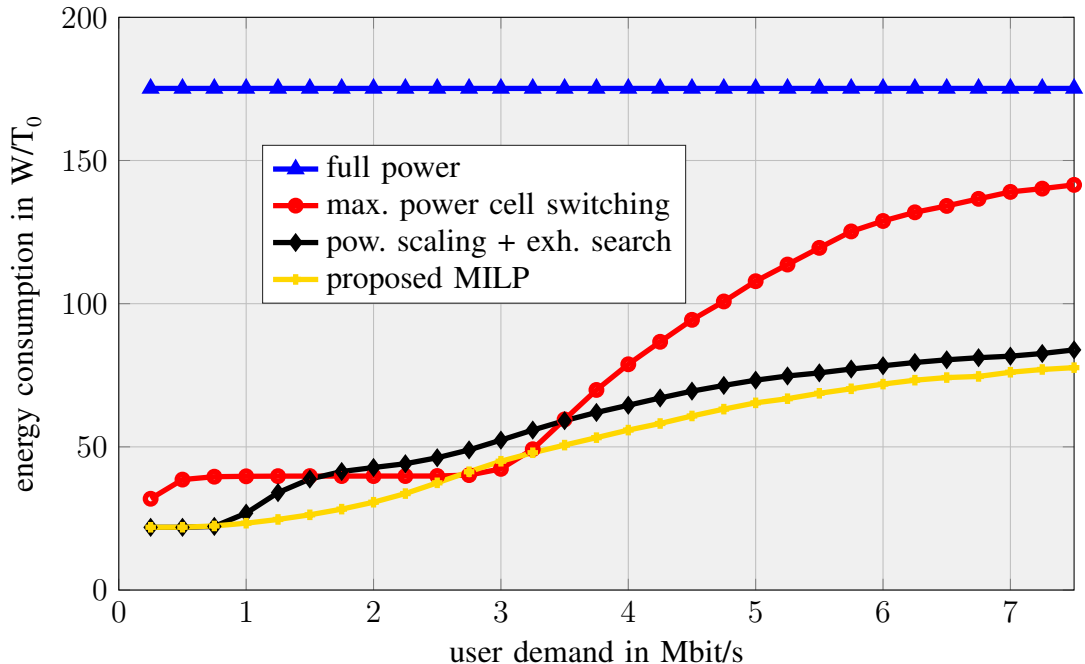

	\centering
	\include{fE_allsolve}
	\caption{Energy consumption for energy minimization schemes over increasing user demand, averaged over 5000 simulations of $M=20$ randomly distributed demand points. The proposed scheme achieves the lowest average energy consumption levels of the evaluated schemes.}
	\label{fig:fullcomp1}
\end{figure}

\begin{figure}
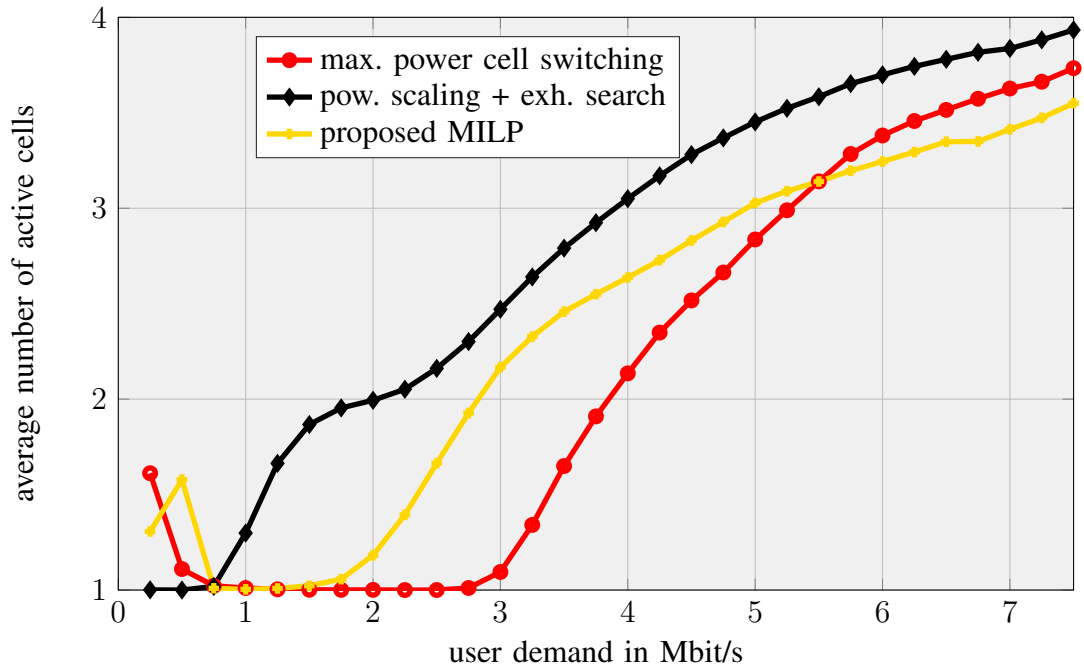

	\centering	
	\include{fN_cellsactive}
	\caption{Number of active cells for energy minimization schemes over increasing user demand, averaged over 5000 simulations of $M=20$ randomly distributed demand points. For high demand, the proposed scheme on average utilizes the lowest number of cells.}
	\label{fig:ncells1}
\end{figure}

In Fig.~\ref{fig:ncells1}, the average number of active cells is shown. For very low demands, it can be observed that the number of cells is not increasing continuously with the demand, as the proposed algorithm for some scenarios serves all users exclusively with pico cells, instead of using a single macro cell. In practice this does not pose a problem since for these low load levels offloading is not required. On average however less than 4 cells are being used, showing that small cells are only used sporadically or for low demand levels. For very high demand levels, the proposed method utilizes the lowest number of cells.

\begin{figure}
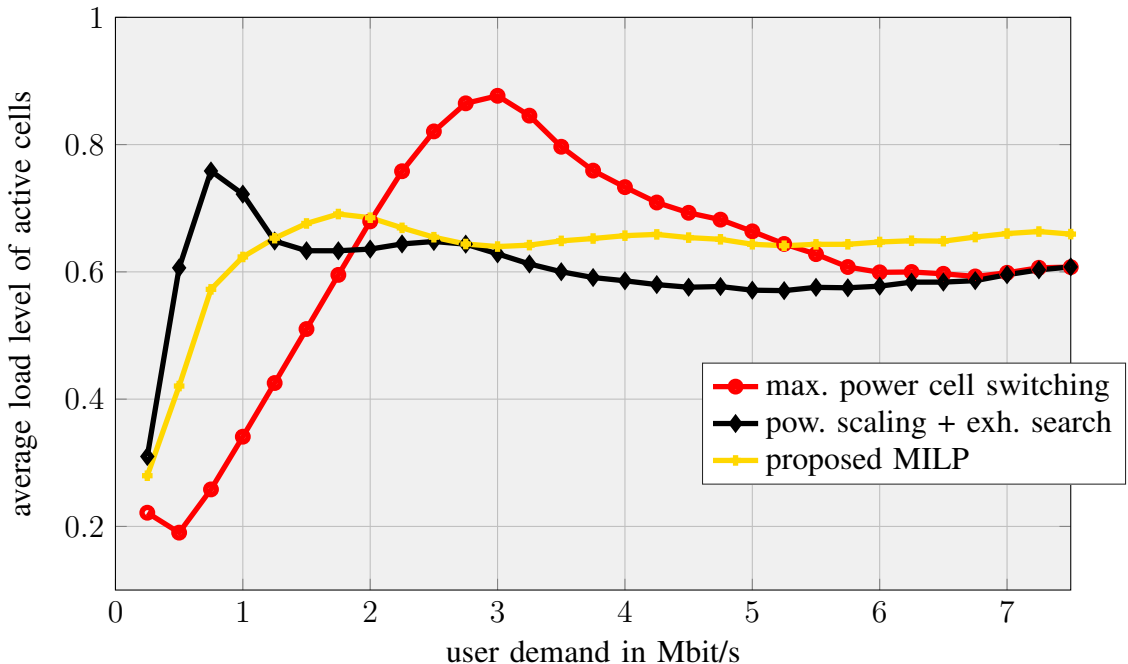

	\centering	
	\include{fL_avgload}
	\caption{Load of active cells for energy minimization schemes over increasing user demand, averaged over 5000 simulations of $M=20$ randomly distributed demand points.}
	\label{fig:avgrho1}
\end{figure}
The average load factor of active cells is shown in Fig.~\ref{fig:avgrho1}. It is observable that the cell load does not converge to 1 even for high loads. It was shown in \cite{yuan15} that for minimum energy consumption, the load would be equal to 1. This however only holds if the transmit power can be increased or decreased without bounds (i.e.~for $P^\text{MIN}=0$ and $P^\text{MIN}=\infty$), and if the cell load is a strictly decreasing function of the transmit power. With the upper- and lower bounds on the transmit power, the discontinuities we introduced in the load computation, and the user allocation changing dynamically with the transmit powers, we observe from Fig.~\ref{fig:avgrho1} that this property no longer holds.

\section{Conclusion} \label{sec:conc}
In this paper we proposed a novel method for minimizing the energy consumption of a wireless communication network, subject to cell load constraints. The transmit power and the cell activity are jointly optimized in a mixed integer linear problem. Multiple simplifications used in other state of the art methods to allow the application of heuristic schemes are not required in the proposed method. \\
The simulation results show that the proposed approach achieved a further decrease in energy consumption relative to both an optimization of the cell activity and a comparable heuristic method. Additionally, it achieves a higher success rate in finding an operable solution for high-demand network scenarios.\\
Even though the proposed method the proposed method consists in linear approximations of the originally mixed integer nonlinear program with bilinear and nonconvex constraints, it still yields very high complexity, making it impractical for the optimization of large networks. Further work could be dedicated to combining existing heuristic methods with an utilization of the proposed approach to optimize smaller clusters of the network, to allow for better scalability.
\appendix
\section{Error constrained load approximation} \label{sec:applin}
As illustrated in Fig.~\ref{fig:linapprox}, we aim to find linear functions 
\begin{equation} \label{eq:urep}
u_i(\gamma) = \alpha_i \gamma + \beta_i,
\end{equation}
indicated with $i=1,\ldots,I$, which satisfy the condition
\begin{equation} \label{eq:ufktrep}
\underset{i}{\max} \; u_i(\gamma) \geq \frac{1}{\log_2(1+\gamma)} \quad \forall \; \gamma \geq \gamma^\mathrm{MIN}.
\end{equation}
with $\alpha_i \leq 0 \ \forall \ i$. The problem of finding suitable linear functions $u_i(\gamma)$ is equivalent to finding a set of breakpoints on $f(\gamma)$ where the linear functions $u_i$ are the lines connecting each two respectively neighboring break points. As discussed in \cite{lin13,lin15}, a good breakpoint selection strategy is to start with the function values of the interval endpoints as the first two breakpoints. Assuming the line connecting these two points to be the linearization solution, we compute the position $\gamma$ where the maximum approximation error occurs. If that error is larger than a predefined threshold $\epsilon$, we add a breakpoint at that position, and we again determine the linear functions between neighboring breakpoints. The procedure is then continued until in each interval between two breakpoints the maximum approximation error is lower than $\epsilon$. \\
Assuming $u(\gamma) \geq f(\gamma)$, we define the approximation error function
\begin{align}
\xi(\gamma) & = u(\gamma) - f(\gamma) \\
& = \alpha \gamma + \beta - \frac{1}{\log_2(\gamma+1)}
\end{align}
and the derivative
\begin{equation}
\frac{\mathrm{d}\xi(\gamma)}{\mathrm{d} \gamma} = \alpha + \frac{\log(2)}{(\gamma + 1) \log^2(\gamma+1)}.
\end{equation}
There is $\frac{\mathrm{d}\xi(\gamma)}{\mathrm{d} \gamma} = 0$ for
\begin{equation}
\gamma = \delta(\alpha_i) = e^{2\mathcal{W} \left( \frac{1}{2} \sqrt{-\frac{\log(2)}{\alpha }}\right)} \quad \forall \; \alpha_i < 0
\end{equation}
where $\mathcal{W}$ is the Lambert W-Function defined as
\begin{equation}
y = f^{-1} \left(y e^y \right) = \mathcal{W} \left(y e^y \right).
\end{equation}
\begin{algorithm} [H]
	\caption{Breakpoint selection algorithm}\label{alg:bps}
	\begin{algorithmic}[1]
		\Procedure{BPS}{$\gamma_1,\gamma_2$}
		\State $\alpha \gets \frac{f(\gamma_2)-f(\gamma_1)}{\gamma_2-\gamma_1}$
		\If{$|\xi(\delta(\alpha))|\leq\epsilon$} 
		\State \textbf{return} $\{\}$ 
		\Else 
		\State \textbf{return} $\{\mathrm{BPS}(\gamma_1,\delta(\alpha)),f(\delta(\alpha)),\mathrm{BPS}(\delta(\alpha),\gamma_2)\}$
		\EndIf
		\EndProcedure
	\end{algorithmic}
\end{algorithm}
To determine the set of breakpoints we define the procedure BPS (Algorithm \ref{alg:bps}) which returns a set of breakpoints necessary between given interval endpoints $(\gamma_1,\gamma_2)$. This way a set of $\gamma$-positions of the breakpoints on $f(\gamma)$ can be obtained, and by connecting each respectively neighboring pair of points in this set, the linear functions $u_i(\gamma)$ can be determined and used in Problem \eqref{eq:powcs}.

\newpage
\bibliographystyle{IEEEtran}
\bibliography{P1B_journal_refs}
\end{document}

%% file: fF_netset.tex
%
%
\definecolor{mycolor1}{rgb}{0.07843,0.16863,0.54902}%
\definecolor{mycolor2}{rgb}{0.84706,0.16078,0.00000}%
\definecolor{mycolor3}{rgb}{0.94118,0.94118,0.94118}%
\begin{tikzpicture}

\begin{axis}[%
width=3in,
height=3in,
at={(0in,0in)},
scale only axis,
xmin=0,
xmax=1000,
xlabel={range (m)},
xmajorgrids,
ymin=0,
ymax=1000,
ylabel={crossrange (m)},
ymajorgrids,
axis background/.style={fill=mycolor3},
legend style={at={(0.325,0.83)},anchor=south west,legend cell align=left,align=left,draw=black}
]
\addplot[only marks,mark=diamond*,mark options={},mark size=8pt,draw=black,fill=mycolor1] plot table[row sep=crcr]{%
200	200\\
150	850\\
800	230\\
780	820\\
};
\addlegendentry{macro BS};

\addplot[only marks,mark=triangle*,mark options={},mark size=6pt,draw=black,fill=mycolor1] plot table[row sep=crcr]{%
500	700\\
520	310\\
320	500\\
690	490\\
};
\addlegendentry{pico BS};

\addplot[only marks,mark=*,mark options={},mark size=3pt,draw=black,fill=mycolor2] plot table[row sep=crcr]{%
838.255587537226	908.10241650695\\
584.71861926332	552.175026715835\\
948.108735396022	32.9398927498766\\
61.0289291925092	53.8629264355561\\
584.641303355111	805.063228558902\\
285.108085658642	451.374854703448\\
827.732173448263	382.646229559959\\
190.986440697398	789.643703689691\\
442.529962202884	364.286869499794\\
393.411506367576	532.34993499891\\
826.573979042765	711.656705981267\\
676.871093438419	871.476517995847\\
207.603034379981	328.689611672229\\
318.104726150263	650.118025397777\\
133.810985356126	974.836148002758\\
671.462889478031	75.9673612941356\\
570.991075462406	587.019167082772\\
169.767066026488	413.88649777336\\
147.655777151737	309.136426466267\\
476.079718267456	263.834041526795\\
};
\addlegendentry{demand point};

\end{axis}
\end{tikzpicture}%

%% file: fP_ratiosolved.tex
%
%
\definecolor{mycolor1}{rgb}{1.00000,0.84314,0.00000}%
\definecolor{mycolor2}{rgb}{0.94118,0.94118,0.94118}%
\begin{tikzpicture}

\begin{axis}[%
width=5in,
height=3in,
at={(0in,0in)},
scale only axis,
xmin=0,
xmax=7.5,
xlabel={user demand in Mbit/s},
xmajorgrids,
ymin=0,
ymax=1,
ylabel={probability of obtaining a feasible solution},
ymajorgrids,
axis background/.style={fill=mycolor2},
legend style={at={(0.176,0.26)},anchor=south west,legend cell align=left,align=left,draw=white!15!black}
]
\addplot [color=red,solid,line width=2.0pt,mark=o,mark options={solid}]
  table[row sep=crcr]{%
0.25	1\\
0.5	1\\
0.75	1\\
1	1\\
1.25	1\\
1.5	1\\
1.75	1\\
2	1\\
2.25	1\\
2.5	1\\
2.75	1\\
3	1\\
3.25	1\\
3.5	0.996875\\
3.75	0.989375\\
4	0.9590625\\
4.25	0.920625\\
4.5	0.8540625\\
4.75	0.7878125\\
5	0.7196875\\
5.25	0.6459375\\
5.5	0.563125\\
5.75	0.49625\\
6	0.43\\
6.25	0.3696875\\
6.5	0.3159375\\
6.75	0.2675\\
7	0.22328125\\
7.25	0.1915625\\
7.5	0.1615625\\
};
\addlegendentry{max. power cell switching};

\addplot [color=black,solid,line width=2.0pt,mark=diamond,mark options={solid}]
  table[row sep=crcr]{%
0.25	1\\
0.5	1\\
0.75	1\\
1	1\\
1.25	1\\
1.5	1\\
1.75	1\\
2	1\\
2.25	0.9996875\\
2.5	0.99796875\\
2.75	0.99140625\\
3	0.97953125\\
3.25	0.96140625\\
3.5	0.92234375\\
3.75	0.87453125\\
4	0.81484375\\
4.25	0.74359375\\
4.5	0.67359375\\
4.75	0.6021875\\
5	0.5346875\\
5.25	0.47296875\\
5.5	0.4096875\\
5.75	0.3478125\\
6	0.29953125\\
6.25	0.25609375\\
6.5	0.213125\\
6.75	0.1859375\\
7	0.16078125\\
7.25	0.1353125\\
7.5	0.1178125\\
};
\addlegendentry{pow. scaling + exh. search};

\addplot [color=mycolor1,solid,line width=2.0pt,mark=+,mark options={solid}]
  table[row sep=crcr]{%
0.25	1\\
0.5	1\\
0.75	1\\
1	1\\
1.25	1\\
1.5	1\\
1.75	1\\
2	1\\
2.25	0.9996875\\
2.5	1\\
2.75	1\\
3	0.99875\\
3.25	0.99734375\\
3.5	0.99390625\\
3.75	0.99015625\\
4	0.98609375\\
4.25	0.978125\\
4.5	0.96671875\\
4.75	0.9540625\\
5	0.9346875\\
5.25	0.90875\\
5.5	0.878125\\
5.75	0.8425\\
6	0.80375\\
6.25	0.75546875\\
6.5	0.705\\
6.75	0.65703125\\
7	0.59609375\\
7.25	0.54046875\\
7.5	0.48953125\\
};
\addlegendentry{proposed MILP};

\end{axis}
\end{tikzpicture}%

%% file: fE_allsolve.tex
%
%
\definecolor{mycolor1}{rgb}{1.00000,0.84314,0.00000}%
\definecolor{mycolor2}{rgb}{0.94118,0.94118,0.94118}%
\begin{tikzpicture}

\begin{axis}[%
width=5in,
height=3in,
at={(0in,0in)},
scale only axis,
xmin=0,
xmax=7.5,
xtick={0,1,2,3,4,5,6,7},
xlabel={user demand in Mbit/s},
xmajorgrids,
ymin=0,
ymax=200,
ytick={0,50,100,150,200},
ylabel={$\text{energy consumption in W/T}_\text{0}$},
ymajorgrids,
axis background/.style={fill=mycolor2},
legend style={at={(0.15,0.5)},anchor=south west,legend cell align=left,align=left,draw=white!15!black}
]
\addplot [color=blue,solid,line width=2.0pt,mark=triangle,mark options={solid}]
  table[row sep=crcr]{%
0.25	175.167155043546\\
0.5	175.167155043546\\
0.75	175.167155043546\\
1	175.167155043546\\
1.25	175.167155043546\\
1.5	175.167155043546\\
1.75	175.167155043546\\
2	175.167155043546\\
2.25	175.167155043546\\
2.5	175.167155043546\\
2.75	175.167155043546\\
3	175.167155043546\\
3.25	175.167155043546\\
3.5	175.167155043546\\
3.75	175.167155043546\\
4	175.167155043546\\
4.25	175.167155043546\\
4.5	175.167155043546\\
4.75	175.167155043546\\
5	175.167155043546\\
5.25	175.167155043546\\
5.5	175.167155043546\\
5.75	175.167155043546\\
6	175.167155043546\\
6.25	175.167155043546\\
6.5	175.167155043546\\
6.75	175.167155043546\\
7	175.167155043546\\
7.25	175.167155043546\\
7.5	175.167155043546\\
};
\addlegendentry{full power};

\addplot [color=red,solid,line width=2.0pt,mark=o,mark options={solid}]
  table[row sep=crcr]{%
0.25	31.8744549239117\\
0.5	38.5685531915361\\
0.75	39.5819419904286\\
1	39.715685021086\\
1.25	39.7747589883294\\
1.5	39.7940168513549\\
1.75	39.8029955650503\\
2	39.8029629796344\\
2.25	39.8107170553472\\
2.5	39.8241213371842\\
2.75	40.1370800390735\\
3	42.3152341789622\\
3.25	49.1749745012463\\
3.5	59.5787394955811\\
3.75	69.8761023315258\\
4	78.8346086082732\\
4.25	86.6894207508603\\
4.5	94.3603969301122\\
4.75	100.772084944923\\
5	107.843272274546\\
5.25	113.663075004895\\
5.5	119.487144949369\\
5.75	125.231737107446\\
6	128.877902793358\\
6.25	131.901690871954\\
6.5	134.138528506326\\
6.75	136.610273343923\\
7	139.031273408682\\
7.25	140.176075861485\\
7.5	141.486674367663\\
};
\addlegendentry{max. power cell switching};

\addplot [color=black,solid,line width=2.0pt,mark=diamond,mark options={solid}]
  table[row sep=crcr]{%
0.25	21.8915144715827\\
0.5	21.8958943804437\\
0.75	22.2171755425387\\
1	26.8478981062892\\
1.25	34.0138592816698\\
1.5	38.7636463808143\\
1.75	41.3820725036307\\
2	42.7920622409683\\
2.25	44.1124625472712\\
2.5	46.1973882007304\\
2.75	48.9364245284415\\
3	52.3572665478821\\
3.25	55.8710594305084\\
3.5	59.1442466360053\\
3.75	62.0120112420905\\
4	64.5308659580881\\
4.25	67.0667125529561\\
4.5	69.461551344073\\
4.75	71.4481259779569\\
5	73.2872398514663\\
5.25	74.7498379892237\\
5.5	75.8347534740292\\
5.75	77.1762696989666\\
6	78.3159264707606\\
6.25	79.4607666044923\\
6.5	80.400559208089\\
6.75	81.1323923850067\\
7	81.6693892669127\\
7.25	82.6034262254524\\
7.5	83.8629641064815\\
};
\addlegendentry{pow. scaling + exh. search};

\addplot [color=mycolor1,solid,line width=2.0pt,mark=+,mark options={solid}]
  table[row sep=crcr]{%
0.25	21.9240969278826\\
0.5	21.9378528911892\\
0.75	22.385238178632\\
1	23.3445471158171\\
1.25	24.6535168306159\\
1.5	26.2704266727102\\
1.75	28.2568610766129\\
2	30.6659818822323\\
2.25	33.6870784894008\\
2.5	37.4429071597945\\
2.75	41.3155135188071\\
3	44.9819281307925\\
3.25	47.9456189271238\\
3.5	50.6128966000133\\
3.75	53.2166048991997\\
4	55.8568935464397\\
4.25	58.1549973547781\\
4.5	60.7887890674152\\
4.75	63.174620508616\\
5	65.3703910202834\\
5.25	66.8226685557283\\
5.5	68.7022406139624\\
5.75	70.2830170300323\\
6	71.9164223437277\\
6.25	73.2817743232338\\
6.5	74.1972579660225\\
6.75	74.593563663145\\
7	76.0531492671726\\
7.25	77.0215098586517\\
7.5	77.6680552553618\\
};
\addlegendentry{proposed MILP};

\end{axis}
\end{tikzpicture}%

%% file: fN_cellsactive.tex
%
%
\definecolor{mycolor1}{rgb}{1.00000,0.84314,0.00000}%
\definecolor{mycolor2}{rgb}{0.94118,0.94118,0.94118}%
\begin{tikzpicture}

\begin{axis}[%
width=5in,
height=3in,
at={(0in,0in)},
scale only axis,
xmin=0,
xmax=7.5,
xlabel={user demand in Mbit/s},
xmajorgrids,
ymin=1,
ymax=4,
ytick={1, 2, 3, 4},
ylabel={average number of active cells},
ymajorgrids,
axis background/.style={fill=mycolor2},
legend style={at={(0.144,0.754)},anchor=south west,legend cell align=left,align=left,draw=white!15!black}
]
\addplot [color=red,solid,line width=2.0pt,mark=o,mark options={solid}]
  table[row sep=crcr]{%
0.25	1.61203419906437\\
0.5	1.11010068203962\\
0.75	1.02220484753714\\
1	1.010625\\
1.25	1.00390625\\
1.5	1.00156298843389\\
1.75	1.00093940817285\\
2	1.00094324791699\\
2.25	1\\
2.5	1.00032573289902\\
2.75	1.01080898496876\\
3	1.09470702779253\\
3.25	1.34103286384977\\
3.5	1.64957610012111\\
3.75	1.90973451327434\\
4	2.13508900268227\\
4.25	2.34856361149111\\
4.5	2.51706117430065\\
4.75	2.66342480027787\\
5	2.83627608346709\\
5.25	2.98928238583411\\
5.5	3.14048257372654\\
5.75	3.28408370323399\\
6	3.38098693759071\\
6.25	3.45688225538972\\
6.5	3.51509250243427\\
6.75	3.57367829021372\\
7	3.62694300518135\\
7.25	3.66268656716418\\
7.5	3.73344651952462\\
};
\addlegendentry{max. power cell switching};

\addplot [color=black,solid,line width=2.0pt,mark=diamond,mark options={solid}]
  table[row sep=crcr]{%
0.25	1.00096789804807\\
0.5	1\\
0.75	1.01907740422205\\
1	1.2984375\\
1.25	1.66359375\\
1.5	1.86714598311973\\
1.75	1.9539689995303\\
2	1.9940260965257\\
2.25	2.05177632627051\\
2.5	2.16107491856678\\
2.75	2.30163823678433\\
3	2.47070277925297\\
3.25	2.63981220657277\\
3.5	2.7908760597497\\
3.75	2.92411504424779\\
4	3.04998780785174\\
4.25	3.16990424076607\\
4.5	3.28097141100523\\
4.75	3.36818339701285\\
5	3.45184590690209\\
5.25	3.5232991612302\\
5.5	3.58391420911528\\
5.75	3.65250475586557\\
6	3.6988388969521\\
6.25	3.74295190713101\\
6.5	3.77994157740993\\
6.75	3.81664791901012\\
7	3.83678756476684\\
7.25	3.88208955223881\\
7.5	3.93378607809847\\
};
\addlegendentry{pow. scaling + exh. search};

\addplot [color=mycolor1,solid,line width=2.0pt,mark=+,mark options={solid}]
  table[row sep=crcr]{%
0.25	1.30666236489756\\
0.5	1.57794738551478\\
0.75	1.00906958561376\\
1	1.0034375\\
1.25	1.0084375\\
1.5	1.02250703344795\\
1.75	1.05714733051511\\
2	1.18220405596604\\
2.25	1.39397801497531\\
2.5	1.66449511400651\\
2.75	1.9270393514609\\
3	2.1671092228713\\
3.25	2.3275117370892\\
3.5	2.45841744045216\\
3.75	2.55\\
4	2.63691782492075\\
4.25	2.72804377564979\\
4.5	2.83092529972333\\
4.75	2.92810003473428\\
5	3.02728731942215\\
5.25	3.08946877912395\\
5.5	3.13994638069705\\
5.75	3.19657577679138\\
6	3.24455732946299\\
6.25	3.29436152570481\\
6.5	3.34858812074002\\
6.75	3.34983127109111\\
7	3.41321243523316\\
7.25	3.47313432835821\\
7.5	3.55008488964346\\
};
\addlegendentry{proposed MILP};

\end{axis}
\end{tikzpicture}%

%% file: fL_avgload.tex
%
%
\definecolor{mycolor1}{rgb}{1.00000,0.84314,0.00000}%
\definecolor{mycolor2}{rgb}{0.94118,0.94118,0.94118}%
\begin{tikzpicture}

\begin{axis}[%
width=5in,
height=3in,
at={(0in,0in)},
scale only axis,
xmin=0,
xmax=7.5,
xlabel={user demand in Mbit/s},
xmajorgrids,
ymin=0.1,
ymax=1,
ylabel={average load level of active cells},
ymajorgrids,
axis background/.style={fill=mycolor2},
legend style={at={(0.614,0.184)},anchor=south west,legend cell align=left,align=left,draw=white!15!black}
]
\addplot [color=red,solid,line width=2.0pt,mark=o,mark options={solid}]
  table[row sep=crcr]{%
0.25	0.221501722132348\\
0.5	0.190207437977149\\
0.75	0.258132865345825\\
1	0.340977966827405\\
1.25	0.425345699149585\\
1.5	0.510178425624253\\
1.75	0.595254181002348\\
2	0.679341889152865\\
2.25	0.758163359605824\\
2.5	0.820746387918203\\
2.75	0.864848293515504\\
3	0.876688765947421\\
3.25	0.845497198577415\\
3.5	0.796611685174892\\
3.75	0.759122334542958\\
4	0.733157382734085\\
4.25	0.70899004704767\\
4.5	0.692838667053076\\
4.75	0.682091162229611\\
5	0.663666345629557\\
5.25	0.643863669310591\\
5.5	0.627868420555801\\
5.75	0.607651544622376\\
6	0.599264062334114\\
6.25	0.599843188460776\\
6.5	0.596968287932952\\
6.75	0.592819115204876\\
7	0.598549141228215\\
7.25	0.606433651572519\\
7.5	0.607466357587047\\
};
\addlegendentry{max. power cell switching};

\addplot [color=black,solid,line width=2.0pt,mark=diamond,mark options={solid}]
  table[row sep=crcr]{%
0.25	0.309890881379022\\
0.5	0.606313066204018\\
0.75	0.758551925147836\\
1	0.722149352529063\\
1.25	0.648450268069789\\
1.5	0.633219953542439\\
1.75	0.633157445245902\\
2	0.635799172206514\\
2.25	0.6438774478678\\
2.5	0.647736090777087\\
2.75	0.643536423023458\\
3	0.628292788496193\\
3.25	0.61245488448152\\
3.5	0.600202429659136\\
3.75	0.591087239762562\\
4	0.585977507583565\\
4.25	0.579998158033531\\
4.5	0.575868742202923\\
4.75	0.57695825732138\\
5	0.571213176857581\\
5.25	0.5705286938876\\
5.5	0.575673610361358\\
5.75	0.574975297791969\\
6	0.577554228295\\
6.25	0.583834204059851\\
6.5	0.583917216354292\\
6.75	0.586182811560145\\
7	0.595505001456623\\
7.25	0.603147504955595\\
7.5	0.607784179911681\\
};
\addlegendentry{pow. scaling + exh. search};

\addplot [color=mycolor1,solid,line width=2.0pt,mark=+,mark options={solid}]
  table[row sep=crcr]{%
0.25	0.279667938667945\\
0.5	0.420569060088205\\
0.75	0.572502371194816\\
1	0.623447916397452\\
1.25	0.652979982677574\\
1.5	0.676064564068241\\
1.75	0.69088684978164\\
2	0.685097303304385\\
2.25	0.669164580722449\\
2.5	0.654442237457837\\
2.75	0.643780036573251\\
3	0.639752114021086\\
3.25	0.642048465309026\\
3.5	0.64870617473038\\
3.75	0.652411065157104\\
4	0.656928326707093\\
4.25	0.658864955006614\\
4.5	0.653849767608454\\
4.75	0.65118021760131\\
5	0.643380113354578\\
5.25	0.641101865030298\\
5.5	0.64327664828213\\
5.75	0.64331803454868\\
6	0.646984359086644\\
6.25	0.648897431498176\\
6.5	0.648240228142949\\
6.75	0.654737841400878\\
7	0.660192798248105\\
7.25	0.66339815490564\\
7.5	0.659356457172308\\
};
\addlegendentry{proposed MILP};

\end{axis}
\end{tikzpicture}%